\documentclass[english,11pt,a4paper,dvipsnames]{article}
\ifx\pdfoutput\undefined\else
\pdfoutput=1
\fi 
\usepackage[T1]{fontenc}
\usepackage[latin9,utf8]{inputenc}
\usepackage{geometry}
\usepackage{babel}
\usepackage{hyperref}
\usepackage[shortlabels]{enumitem}
\usepackage{float}
\usepackage{stmaryrd}

\usepackage{amsmath, amssymb, amstext, mathtools}
\usepackage{amsthm}

\usepackage{algorithm}
\usepackage{algpseudocodex}
\usepackage{tikz}
\usepackage{pgfplots}
\usepackage{textpos}

\usepackage[capitalise]{cleveref}

\usepackage{xcolor}

\usetikzlibrary{fadings,backgrounds,patterns,arrows,decorations.pathreplacing,decorations.pathmorphing,calc}

\newtheorem{theorem}{Theorem}[section]
\newtheorem{lemma}[theorem]{Lemma}
\newtheorem{corollary}[theorem]{Corollary}
\newtheorem{proposition}[theorem]{Proposition}

\newtheorem{conjecture}[theorem]{Conjecture}
\newtheorem{claim}[theorem]{Claim}

\theoremstyle{definition}
\newtheorem{definition}[theorem]{Definition}

\theoremstyle{remark}
\newtheorem{remark}[theorem]{Remark}

\newenvironment{claimproof}{\begin{proof}}{\end{proof}}

\newcommand{\CS}{{\mathcal S}}

\newcommand{\NN}{{\mathbb N}}

\newcommand{\IS}{\mathrm{IS}}

\newcommand{\numindsub}[2]{\#\mathrm{IndSub}(#1 \to #2)}
\newcommand{\numindsubstar}[1]{\#\mathrm{IndSub}(#1 \to \,\star\,)}

\newcommand{\numsubstar}[1]{\#\mathrm{Sub}(#1 \to \,\star\,)}

\newcommand{\pindsub}[1]{\#\textnormal{\textsc{IndSub}}(#1)}

\newcommand{\slice}[2]{#1^{(#2)}}

\newcommand{\scorp}[1]{\Psi_{#1}}
\newcommand{\altEnum}[2]{\widehat{#1}(#2)}

\newcommand{\sinkprop}{\Psi_{\mathrm{sink}}}

\newcommand{\fp}[2]{\mathrm{FP}(#1, #2)}

\DeclareMathOperator{\set}{set}

\DeclareMathOperator{\Aut}{Aut}

\newcommand{\sharpwone}{\textnormal{\textsf{\#W[1]}}}

\newcommand{\sharpP}{\textnormal{\textsf{\#P}}}
\newcommand{\ETH}{\textnormal{\textsf{ETH}}}
\newcommand{\SETH}{\textnormal{\textsf{SETH}}}

\usepackage[textwidth=2.5cm]{todonotes}

\newcommand{\orcid}[1]{\href{https://orcid.org/#1}{\includegraphics[height=1.8ex]{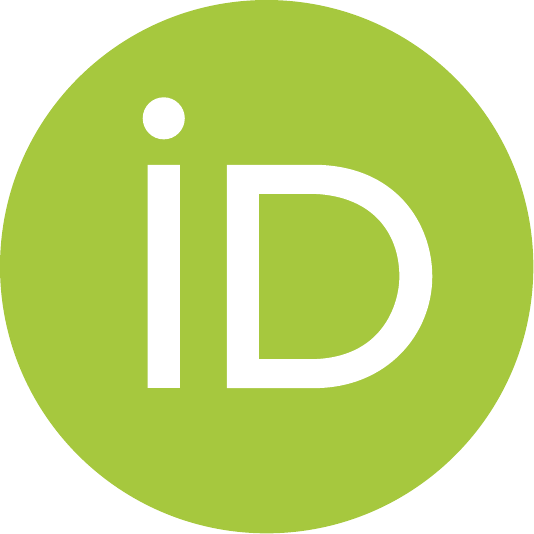}}}

\makeatletter
\def\blfootnote{\gdef\@thefnmark{}\@footnotetext}
\makeatother

\title{Counting Small Induced Subgraphs: \\Scorpions Are Easy but Not Trivial%
\blfootnote{\rightskip=5.7cm
The research is funded by the European Union (ERC, CountHom, 101077083).
Views and opinions expressed are those of the author(s) only and do not necessarily reflect those of the European Union or the European Research Council Executive Agency. Neither the European Union nor the granting authority can be held responsible for them.
}
}
\author{
Radu Curticapean \orcid{0000-0001-7201-9905} \\
University of Regensburg and IT University of Copenhagen
\and
Simon D\"{o}ring \orcid{0009-0002-6667-5257} \\
Max Planck Institute for Informatics and Saarland University
\and
Daniel Neuen \orcid{0000-0002-4940-0318}\\
Max Planck Institute for Informatics
}
\date{}

\definecolor[named]{urlblue}{cmyk}{1,0.58,0,0.21}
\hypersetup{
    breaklinks=true,
    colorlinks=true,
    citecolor=purple!70!blue!60!black,
    linkcolor=purple!70!blue!90!black,
    urlcolor=urlblue,
    pdflang={en},
    pdftitle={Counting Small Induced Subgraphs: Scorpions Are Easy but Not Trivial},
    pdfauthor={Radu Curticapean, Simon Döring, Daniel Neuen}
}

\begin{document}

\maketitle
\begin{textblock}{5}(7.85, 8.05) \includegraphics[width=150px]{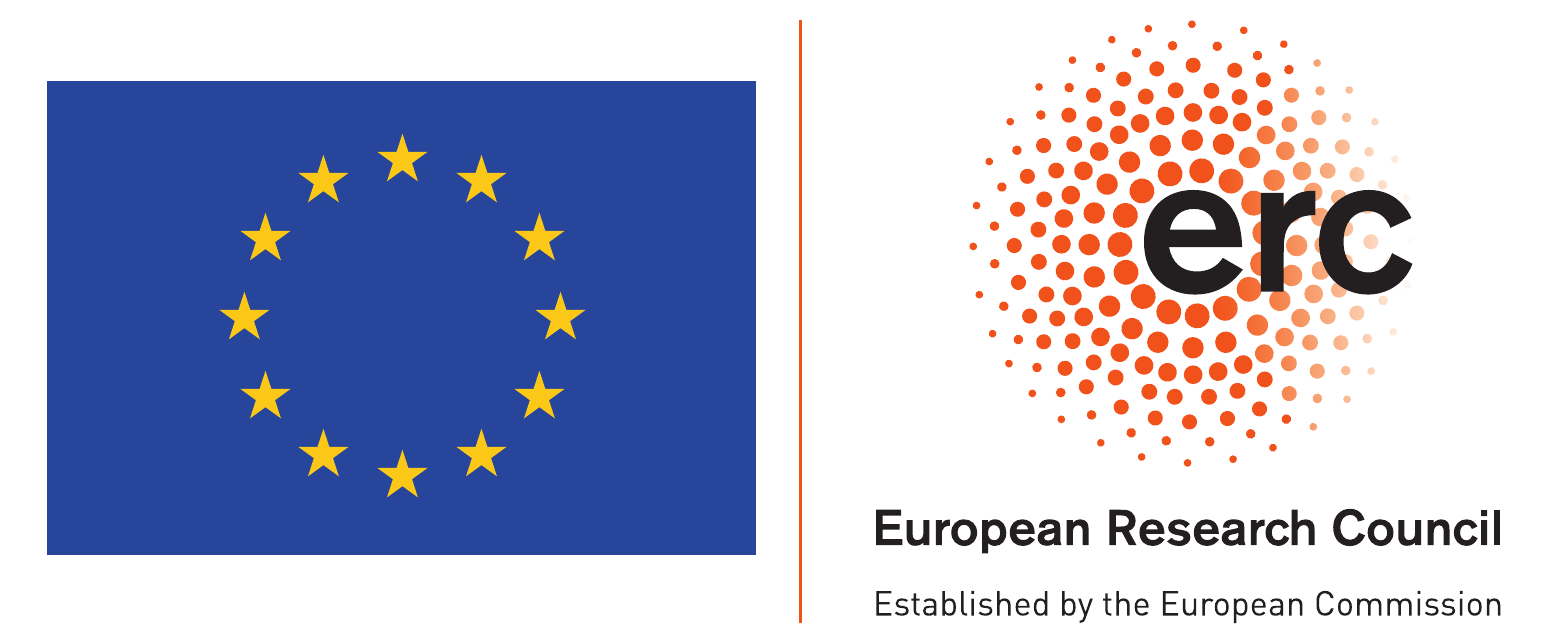} \end{textblock}

\begin{abstract}
    We consider the parameterized problem $\pindsub{\Phi}$ for fixed graph properties $\Phi$: Given a graph $G$ and an integer $k$, this problem asks to count the number of induced $k$-vertex subgraphs satisfying $\Phi$.
    D{\"{o}}rfler et al.\ [Algorithmica 2022] and Roth et al.\ [SICOMP 2024]  conjectured that $\pindsub{\Phi}$ is $\sharpwone$-hard for all non-meager properties $\Phi$, i.e., properties that are nontrivial for infinitely many $k$.
    This conjecture has been confirmed for several restricted types of properties, including all hereditary properties~[STOC 2022] and all edge-monotone properties~[STOC 2024].

    In this work, we refute this conjecture by showing that scorpion graphs, certain $k$-vertex graphs which were introduced more than 50 years ago in the context of the evasiveness conjecture, can be counted in time $O(n^4)$ for all $k$.
    A simple variant of this construction results in graph properties that achieve arbitrary intermediate complexity assuming $\ETH$.

    We formulate an updated conjecture on the complexity of $\pindsub{\Phi}$ that correctly captures the complexity status of scorpions and related constructions.
\end{abstract}

\clearpage

\section{Introduction}

Counting small patterns in graphs is a fundamental problem in computer science, with applications in bioinformatics \cite{DBLP:journals/tcsb/SchreiberS05}, network analysis \cite{Milo02, SchillerJHS15}, databases \cite{DBLP:conf/stoc/GroheSS01}, and other areas.
In this paper, we focus on generalizations of counting induced $H$-copies in a large $n$-vertex input graph $G$, for a fixed $k$-vertex graph $H$.

\paragraph{Counting Induced $H$-Copies.}
The starting point of our investigation is the problem of counting induced $H$-copies for a fixed \emph{individual} graph $H$.
Formally, this problem asks for the number of sets $X \subseteq V(G)$ such that the induced subgraph $G[X]$ is isomorphic to $H$.
We stress that $H$ is considered fixed in this problem, and only $G$ is the input.
In particular, each such problem can be solved in time $O(n^k)$, which is polynomial in $n$ for fixed $k$. 

Some improvements over the trivial $O(n^k)$ running time are known:
For example, triangles can be counted in $O(n^\omega)$ time \cite{DBLP:journals/siamcomp/ItaiR78}, where $\omega < 2.372$ is the optimal exponent of $n \times n$ matrix multiplication \cite{DBLP:conf/soda/AlmanDWXXZ25}. 
This can be lifted to cliques beyond $K_3$: Denoting by $C(K_k)$ the optimal exponent for counting $k$-cliques in $n$-vertex graphs, similar to the exponent of matrix multiplication, we have $C(K_k) \leq \omega \cdot \lceil k/3 \rceil$ (see \cite{NesetrilP85}).
Under the Exponential-Time Hypothesis $\ETH$, there exists a fixed constant $\alpha$ such that $C(K_k)\geq \alpha \cdot k$ (see \cite{ChenHKX06}).

Writing $C_\mathrm{ind}(H)$ to denote the optimal exponent of counting induced $H$-subgraphs~\cite{CurticapeanDM17}
for fixed $H$, a straightforward reduction shows that $C_\mathrm{ind}(H) = C(K_k)$ for all $k$-vertex graphs $H$.
In other words, for fixed $k$, all induced $H$-counting problems with $k$-vertex $H$ are equally hard, and they require an exponent of $\Omega(k)$ under $\ETH$.
Compare this to counting not necessarily induced subgraphs, where different patterns can yield different complexity exponents: The number of subgraph copies of the edgeless $k$-vertex graph $\IS_k$ is always~${n \choose k}$ and can thus be computed in linear time, while counting $k$-cliques is hard.

\paragraph{Counting Patterns From a Set $\mathcal H$.}
In recent years, counting occurrences of individual $k$-vertex patterns $H$ has been generalized to counting pattern occurrences from a fixed \emph{set} of patterns~\cite{JerrumM15,JerrumM15b}.
In this setting, we fix a number $k \in \mathbb N$ and a set $\mathcal H$ of $k$-vertex graphs.
On input $G$, we wish to count the induced $k$-vertex subgraphs of $G$ isomorphic to some $H\in \mathcal H$.
This subsumes the problem of counting induced $H$-copies, but also allows us to address, e.g., the problem of counting connected $k$-vertex graphs \cite{JerrumM15}.

Of course, every such problem can be solved by counting induced $H$-copies for the individual graphs $H\in \mathcal H$, which readily implies an $O(n^{C(K_k)})$ time algorithm for counting induced copies from $\mathcal H$.
Significantly faster algorithms were only known for \emph{trivial} pattern sets $\mathcal H$, i.e., if the pattern set $\mathcal H$ is empty or contains all $k$-vertex graphs. In these cases, the output is just $0$ or ${n \choose k}$, respectively.
More specifically, no nontrivial set $\mathcal H$ of $k$-vertex graphs with exponent strictly less than $C(K_{\lfloor k/2\rfloor})$ was known (see, e.g., \cite{CurticapeanN25}).
In other words, counting induced patterns from a fixed set $\mathcal H$ of $k$-vertex graphs appeared to be either trivial or very hard.

\paragraph{Parameterized Complexity.}

In the literature, pattern counting problems are often phrased in terms of graph properties $\Phi$ that may hold on infinitely many graphs rather than finite sets $\mathcal H$: In the problem $\pindsub{\Phi}$ for a fixed graph property $\Phi$, the input is a graph $G$ and $k \in \mathbb N$, and we ask to count the induced $k$-vertex subgraphs of $G$ satisfying $\Phi$. Compared to the previous setting, the pattern size $k$ is now part of the input.
We say that $\pindsub{\Phi}$ is \emph{fixed-parameter tractable} if it can be solved in time $f(k) \cdot n^{O(1)}$ for a computable function $f$, and we would like to understand which properties $\Phi$ render $\pindsub{\Phi}$ fixed-parameter tractable.

In prior literature, all known properties $\Phi$ with fixed-parameter tractable $\pindsub{\Phi}$ are essentially trivial:
Formally, we call $\Phi$ \emph{meager} if the restriction $\Phi^{(k)}$ of $\Phi$ to $k$-vertex graphs (i.e., the $k$-th slice of $\Phi$) is trivial for all but finitely many $k\in \mathbb N$.
Meager computable properties $\Phi$ trivially render $\pindsub{\Phi}$ fixed-parameter tractable.
Complementing this, the problem $\pindsub{\Phi}$ was conjectured to be $\sharpwone$-hard for all other computable properties $\Phi$. Here, $\sharpwone$ is the parameterized analogue of $\sharpP$; it is known that $\ETH$ rules out fixed-parameter tractable algorithms for all $\sharpwone$-hard problems~(see, e.g., \cite{FlumG06}).

\begin{conjecture}[\cite{DorflerRSW22,FockeR24,RothSW24}]
    \label{conj:non-meager-hard}
    For every property $\Phi$ that is computable and not meager, the problem $\pindsub{\Phi}$ is $\sharpwone$-hard.
\end{conjecture}
This conjecture has been verified for wide ranges of properties $\Phi$, e.g., properties that are closed under deleting edges \cite{CurticapeanN25, DoringMW24, DoringMW25} or deleting vertices \cite{FockeR24}, and various other natural classes of properties \cite{CurticapeanN25,JerrumM15, JerrumM15b, RothS20, RothSW24}.
This made it a plausible working hypothesis in the area.

\paragraph{Our Results.}
We show that \textbf{Conjecture~\ref{conj:non-meager-hard} fails}. That is, we exhibit non-meager properties $\Phi$ such that $\pindsub{\Phi}$ is fixed-parameter tractable---in fact, it is even polynomial-time solvable.
These counterexamples are derived from simple constructions that were introduced 50 years ago in the context of the \emph{evasiveness conjecture}~\cite{LenstraBE74, DBLP:journals/sigact/Rosenberg73,Kozlov08}, a yet unresolved major open problem about the worst-case query complexity of graph properties.

\begin{itemize}
    \item As a simple counterexample, the property $\sinkprop$ of having a sink vertex in a directed graph is nontrivial on graphs of fixed size $k\geq 2$, and therefore not meager. Nevertheless, $\pindsub{\sinkprop}$ can be solved in linear time on directed graphs without antiparallel edges. 
    We invite the reader to discover the algorithm themselves before proceeding to Section~\ref{sec:sink}. Sinks also presented the first counterexample to a (too strong) version of the evasiveness conjecture on directed graphs~\cite{DBLP:journals/sigact/Rosenberg73}.
    \item A marginally more involved construction also works for undirected graphs: The \emph{scorpion property} $\Psi$ is a non-meager property of undirected graphs such that $\pindsub{\Psi}$ can be solved in $O(n^4)$ time on general undirected graphs. Scorpions were the first counterexample to a (too strong) variant of the evasiveness conjecture for undirected graphs \cite{LenstraBE74}, prompting the restriction of this conjecture to monotone properties.
    \item We show more generally that $\pindsub{\Psi}$ can be made gradually harder: 
    For every $\ell \in \mathbb N$, we construct a \emph{generalized scorpion property} $\scorp{\ell}$ such that $\pindsub{v}$ can be solved in $O(n^{\ell+3})$ time,
    while $\ETH$ rules out $O(n^{\alpha \cdot \ell})$ time algorithms for a fixed constant $\alpha >0$.
    Our construction also allows for $\ell$ to be a function in $k$.
    Under the Strong Exponential-Time Hypothesis $\SETH$, we rule out $O(n^{\ell + 2 - \varepsilon})$ time algorithms.
\end{itemize}
Finally, informed by these counterexamples, we formulate a new hypothesis on the computational complexity of $\pindsub{\Phi}$ for properties $\Phi$.
This new hypothesis is more technical and explained in Section~\ref{sec:conj}.
In a nutshell, it is based around the well-established fact that sums of induced pattern counts like $\numindsub{\Phi^{(k)}}{G}$ can be expressed as linear combinations of (not necessarily induced) subgraph counts \cite{CurticapeanDM17}, and that such basis changes may help in understanding the complexity of a problem~\cite{CurticapeanDM17,CurticapeanDN025,CurticapeanN25,DorflerRSW22,DoringMW24,DoringMW25,FockeR24,GoldbergR24,DBLP:conf/icml/JinBCL24,RothS20,RothSW24,DBLP:conf/icalp/Roth0W21}.
More specifically, the new hypothesis postulates that a useful phenomenon occurs when expressing the graph parameter $\numindsub{\Phi^{(k)}}{G}$ as a linear combination of subgraph counts: Any hard term in such a linear combination ensures hardness of the entire linear combination. As outlined in Section~\ref{sec:conj},  general linear combinations of $k$-vertex subgraph counts do not enjoy this useful phenomenon.

\begin{figure}[t]
    \centering
    \includegraphics[width=0.65\linewidth]{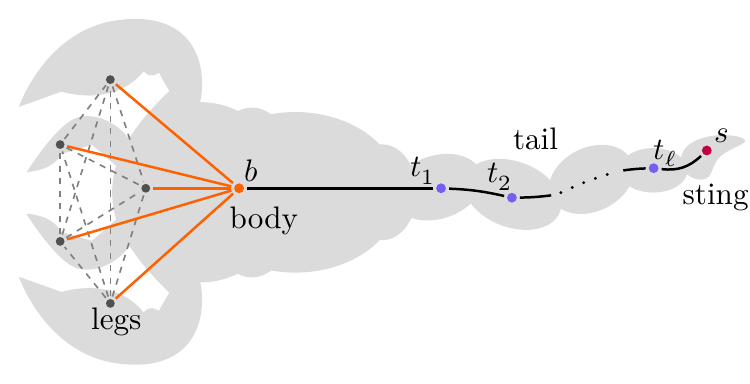}
    \caption{
    A graph $H$ is an $\ell$-scorpion if it has the above form:
    Dashed edges \emph{may} be present in $H$ or not, solid edges \emph{must} be present, and non-drawn edges \emph{must not} be present.
    }
    \label{fig:scorpion}
\end{figure}

\section{Preliminaries}

We write $\NN = \{1,2,3,\dots\}$ and $[n] \coloneqq \{1, ..., n\}$ for $n \in \NN$. 
For a set $A$, we write $\binom{A}{k}$ for the set of all $k$-element subsets of $A$.

\paragraph{Graph Theory.}
We follow standard textbooks~\cite{Diestel17} for graph-theoretic notation.
Unless stated otherwise, graphs are simple (i.e., without multiedges or self-loops) and undirected.
We write $V(G)$ and $E(G)$ for the vertex and edge set of $G$, respectively, and we write $N_G(v) \coloneqq \{u \mid uv \in E(G)\}$ for the neighborhood of $v \in V(G)$ in $G$.

A graph $H$ is a \emph{subgraph} of $G$, written $H \subseteq G$, if it can be obtained by deleting vertices and edges from $G$.
For $X \subseteq V(G)$, we write $G[X]$ to denote the \emph{induced subgraph} on $X$.
For $S \subseteq E(G)$, we write $G[S]$ for the subgraph with vertex set $V(G)$ and edge set $S$. 
For $k,\ell \in \NN$, we write $K_\ell$ for the complete $\ell$-vertex graph, $\IS_\ell$ for the edgeless $\ell$-vertex graph, and $K_{\ell,k}$ for the complete bipartite graph on $\ell+k$ vertices.

\paragraph{Induced Subgraph Counts.}
A \emph{graph property} $\Phi$ is a function that maps each graph $G$ to $\{0, 1\}$ and is invariant under isomorphisms, i.e., $\Phi(G) = \Phi(H)$ for all isomorphic graphs $G,H$.
For $k \in \NN$, the \emph{$k$-th slice of $\Phi$}, denoted by $\slice{\Phi}{k}$, is the restriction of $\Phi$ to $k$-vertex graphs.
We implicitly identify $\Phi^{(k)}$ with the set of all $k$-vertex graphs $G$ satisfying $\Phi(G) = 1$, and we use graph properties and sets of graphs interchangeably.
For $k \in \NN$ and a graph $G$, the number of $k$-vertex induced subgraphs $G[X]$ satisfying $\Phi$ will be denoted by
\[\numindsub{\Phi^{(k)}}{G} \coloneqq \sum_{\substack{X \subseteq V(G)\\|X| = k}} \Phi(G[X]).\]
Moreover, we write $\numindsubstar{\Phi^{(k)}}$ for the map $G \mapsto \numindsub{\Phi^{(k)}}{G}$.

\paragraph{Complexity Theory.}
A \emph{parameterized problem} consists of a function $P \colon \Sigma^\ast \to \NN$ and a computable parameterization $\kappa \colon \Sigma^\ast \to \NN$. It is \emph{fixed-parameter tractable} (FPT) if there is a computable function $f$, a constant $c \in \mathbb N$, and a deterministic algorithm $\mathbb{A}$ that computes $P(x)$ in time $O(f(\kappa(x))\cdot |x|^c)$ for all $x \in \Sigma^\ast$. We write $\pindsub{\Phi}$ for the parametrized problem that gets as input a graph $G$ and a \emph{parameter} $k$, and computes $\numindsub{\Phi^ {(k)}}{G}$.

For lower bounds, we rely on the \emph{Exponential-Time Hypothesis} ($\ETH$)~\cite{IPZ01}, which asserts the existence of some $\varepsilon > 0$ such that the Boolean satisfiability problem on $n$-variable $3$-CNF formulas cannot be solved in time $O(2^{\varepsilon \cdot n})$ (see also \cite[Conjecture 14.1]{CyganFKLMPPS15}). The \emph{Strong Exponential-Time Hypothesis} ($\SETH$) \cite{IP01} states that for all $\varepsilon > 0$, there is a $k \geq 3$ such that Boolean Satisfiability on $k$-CNF formulas cannot be solved in $O(2^{(1 - \varepsilon) \cdot n})$ (see also \cite[Conjecture 14.2]{CyganFKLMPPS15}). 

\section{Main Result}

To present the idea underlying the tractability of scorpions, we first consider a variant for directed graphs, where this idea becomes particularly simple.
Then we introduce scorpions and their generalizations and prove the claimed upper and lower complexity bounds.
Finally, we observe that scorpions show that some known complexity lower bounds are tight.

\subsection{Directed Graphs Containing a Sink}
\label{sec:sink}
As a warm-up, we consider a property $\sinkprop$ of directed graphs that is nontrivial but yields a linear-time counting problem $\pindsub{\sinkprop}$.
For this subsection, we momentarily consider directed graphs without antiparallel edges, i.e., at most one of the edges $(u,v)$ and $(v,u)$ may be present.
The property $\sinkprop$ is defined to hold on $H$ if there is a \emph{sink} vertex $s\in V(H)$, i.e., a vertex $s$ such that $(u,s) \in E(H)$ for all $u \in V(H)\setminus\{s\}$.
This property is clearly nontrivial in every slice $k \in \mathbb N$, as there are $k$-vertex graphs with a sink and $k$-vertex graphs without a sink (e.g., an in-star versus the edgeless graph).

Towards an algorithm, we observe crucially that every graph without antiparallel edges contains at most one sink, since two distinct sinks $u,v$ would imply the presence of both edges $(u,v)$ and $(v,u)$.
Hence, the set of $k$-vertex sets containing a sink
\[\mathcal X \coloneqq \Big\{X \in \binom{V(G)}{k} ~\Big|~  G[X] \in \sinkprop\Big\}\]
can be partitioned, according to the unique sink, into
\[
\mathcal X = \bigcup_{v\in V(G)}\mathcal X_v \quad\text{with}\quad
\mathcal X_v \coloneqq \{X \in \mathcal X \mid v\text{ is the sink of }G[X]\}.\]
Finally, for fixed $v\in V(G)$, every set $X \in \mathcal X_v$ has the form $X = \{v,w_1,\ldots, w_{k-1}\}$ with all $w_1,\ldots, w_{k-1}$ pairwise distinct, distinct from $v$, and incoming neighbors of $v$, i.e., they satisfy $(w_i,v) \in E(G)$. 
Writing $\mathrm{in}_G(v)$ for the number of incoming neighbors of $v$, it follows that 
\[
|\mathcal X_v|= {\mathrm{in}_G(v) \choose {k-1}}.
\]
Combining the above equations, we readily obtain a linear-time algorithm for $\pindsub{\sinkprop}$ by computing $\mathrm{in}_G(v)$ for all $v\in V(G)$ and evaluating the resulting formula
\begin{equation}
    \label{eq:indsub:sink}
    \numindsub{\slice{\sinkprop}{k}}{G} = \sum_{v \in V(G)} |\mathcal{X}_v| = \sum_{v\in V(G)}{\mathrm{in}_G(v) \choose {k-1}}.
\end{equation}

\subsection{Generalized Scorpions}

The algorithmic idea for counting $k$-vertex graphs with a sink applies whenever the set $\mathcal X$ of induced subgraphs to be counted admits a partition into few sets $\mathcal X_i$ such that each $|\mathcal X_i|$ is easily determined.
More specifically, we consider partitions in which each $\mathcal X_i$ is determined by the manifestation of a special small set of uniquely identifiable vertices (e.g., the sink vertex), while the subgraph induced by the other vertices is irrelevant.

To apply this idea to undirected graphs,
we use a construction that was first presented in \cite{LenstraBE74} for the so-called \emph{evasiveness conjecture}, and which has become a standard example in this context (see, e.g., \cite[Section 13.1]{Kozlov08}):
A graph $H$ is a \emph{scorpion} if 
it can be obtained from an arbitrary graph $H'$ with $|V(H')|\geq 2$ by adding fresh vertices $b,t,s$,
making $b$ adjacent to all of $H'$, and then adding the edges $bt$ and $ts$.
The vertices $b,t,s$ are usually called \emph{body}, \emph{tail} and \emph{sting},
and it can be shown crucially (see Lemma \ref{lem:scorpion:unique} below) that these vertices are uniquely recoverable from $H$.
Similarly to the arguments above, this allows us (in Theorem~\ref{thm:alg-scorpion}) to design an efficient algorithm to count induced scorpions.

In this paper, we prove these statements for a slightly generalized version of scorpions, since this allows us to obtain gradually harder properties. 
Towards this end, we replace the tail vertex $t$ by a path of $\ell$ vertices $t_1,\dots,t_\ell$. See also Figure~\ref{fig:scorpion}.
\begin{definition}
\label{def:scorpion}
    For $\ell \in \mathbb{N}$, an \emph{$\ell$-scorpion} is a graph $H$ with $|V(H)|\geq \ell +4$ that admits a tuple of pairwise distinct vertices
    \[
(\ \underbrace{\hspace{5pt}b\hspace{5pt}}_{\text{body of }H},\ 
\underbrace{t_1, \dots, t_\ell}_{\text{tail of }H},\ 
\underbrace{\hspace{5pt}s\hspace{5pt}}_{\text{sting of }H}\ ) \in V(H)^{\ell+2},
\]
such that the following holds: Writing $Q:=\{b,t_1,\dots,t_\ell,s\}$ and calling the vertices in $V(H) \setminus Q$ the \emph{legs} of $H$, we have that 
    \begin{itemize}
        \item the graph $H[Q]$ is an induced path from the body $b$ to the sting $s$,
        \item the body $b$ is adjacent to all legs, and
        \item the body $b$ is the only vertex in $Q$ adjacent to legs.
    \end{itemize} 
    We define $\scorp{\ell}$ as the class of all $\ell$-scorpions.
\end{definition}

The property $\Psi_\ell$ is non-meager for all $\ell \geq 1$:
Indeed, for $k \geq \ell + 4$, at least one $k$-vertex scorpion exists (e.g., consisting of the sting, tail, body, and an independent set of legs), while the $k$-vertex graph $K_k$ is not a scorpion.

Note that scorpion graphs are $1$-scorpions.
Also note that the definition speaks about ``the'' body, tail, and sting vertices, as if they were unique. Indeed, they are:

\begin{lemma} \label{lem:scorpion:unique}
    If $H$ is an $\ell$-scorpion, then its body, tail, and sting are unique.
\end{lemma}
\begin{proof}
    Let $k = |V(H)|$ and recall that $k \geq \ell +4$ by Definition~\ref{def:scorpion}.
    The body has degree $k - \ell - 1 \geq 3$.
    Since every leg has degree at most $k - \ell - 2$, we conclude that $H$ contains exactly one vertex of degree $k - \ell - 1$; this uniquely identifies $b$.

    The vertices not adjacent to $b$ are precisely $t_2,\ldots,t_\ell,s$.
    Among these, $s$ is the only vertex of degree $1$ and is thus uniquely identified.
    The vertices $t_1,\dots,t_\ell$ are uniquely identified by their distance from $s$, since vertex $t_i$ is the only vertex at distance $\ell-i+1$ from $s$.
\end{proof}

We use this lemma to show that $\pindsub{\Psi_\ell}$ can be solved in polynomial time for every fixed $\ell \geq 1$, similarly to the algorithm presented in Section~\ref{sec:sink}.

\begin{theorem}
    \label{thm:alg-scorpion}
    There is an algorithm that,
    given $k \geq \ell + 4$ and an $n$-vertex graph $G$ as input, 
    computes $\numindsub{\slice{\scorp{\ell}}{k}}{G}$ in time $O(\ell \cdot n^{\ell + 3})$.
\end{theorem}

\begin{proof}
    Let $\mathcal X$ be the set of induced $k$-vertex graphs in $G$ that are $\ell$-scorpions.
    In our proof, we use the uniqueness of body, tail, and sting to partition $\mathcal X$ into classes $\mathcal X_\mathbf q$ for $q \in V(G)^{\ell+2}$ and then determine each $|\mathcal X_q|$ in linear time. Then the algorithm follows from
    \begin{equation}
    \label{eq:partition-scorpions}
        \numindsub{\slice{\scorp{\ell}}{k}}{G} = |\mathcal X| = \sum_{\mathbf q \in V(G)^{\ell+2}} |\mathcal X_\mathbf q|.
    \end{equation}
    
    More specifically, let $\mathcal P\subseteq V(G)^{\ell+2}$ denote the tuples inducing a path in $G$. Given $\mathbf{q} \in \mathcal P$ with $\mathbf{q} = (b, t_1, \dots, t_\ell, s)$, let $\mathcal X_\mathbf q \subseteq \mathcal X$ denote the set of $k$-vertex $\ell$-scorpions in $G$ with body $b$, tail $t_1,\ldots,t_\ell$ and sting $s$.
    By Lemma~\ref{lem:scorpion:unique}, the sets $\mathcal X_\mathbf q$ for $\mathbf q \in \mathcal P$ partition $\mathcal X$, so \eqref{eq:partition-scorpions} holds with $\mathcal X_\mathbf q = \emptyset$ for $\mathbf q \notin \mathcal P$.
    It remains to determine $|\mathcal X_\mathbf q|$ for $\mathbf{q} = (b, t_1, \dots, t_\ell, s) \in \mathcal P$.
    We show that 
    \begin{equation}
    \label{eq:partition-scorpion-classsize}
        |\mathcal X_\mathbf q| = \binom{|X_G(\mathbf q)|}{k-\ell+2}
        \quad \text{with}\ 
        X_G(\mathbf q) \coloneqq \Biggl | N_G(b) \setminus \bigcup_{v\in\{t_1,\ldots,t_\ell ,s\}}N_G(v) \Biggr|.
    \end{equation}
    
    Indeed, since the $\ell+2$ vertices in $\mathbf q$ induce a path, the set $\mathcal X_\mathbf q$ consists of all $k$-vertex sets $X$ that contain $\{b, t_1, \dots, t_\ell, s\}$ and $k-\ell-2$ additional vertices (the legs) that are adjacent to $b$ and not adjacent to any $t_i$ or $s$.
    The number of such sets $X$ is precisely $|\mathcal X_\mathbf q|= {X_G(\mathbf q)\choose k-\ell+2}$.

    An algorithm with the claimed running time follows from evaluating \eqref{eq:partition-scorpions} term by term while using \eqref{eq:partition-scorpion-classsize} to determine $|\mathcal X_\mathbf q|$. Indeed, $\mathcal P$ can be enumerated in time $O(n^{\ell+2})$ and $|\mathcal X_\mathbf q|$ for $\mathbf q \in \mathcal P$ can be computed via \eqref{eq:partition-scorpion-classsize} in time $O(\ell n)$.
\end{proof}

Since $\Psi_\ell$ is non-meager for every $\ell \geq 1$, Theorem~\ref{thm:alg-scorpion} refutes Conjecture~\ref{conj:non-meager-hard} even with $\ell = 1$.

\subsection{Lower Bounds Based on ETH}

Next, we show that the running time obtained in \Cref{thm:alg-scorpion} is essentially optimal under $\ETH$.
In particular, by choosing $\ell$ as a function of $k$, we can use $\ell$-scorpions to obtain properties with varying computational difficulty.

To obtain the lower bound, we rely on a result from \cite{CurticapeanDN025,CurticapeanN25} that provides a lower bound based on the set of possible Hamming weights attained by a property $\Phi$:
Given a graph property $\Phi$ and $k \in \NN$, we say that $\Phi$ \emph{attains weight $\ell$ on slice $k$} if there is a graph $H\in\Phi$ with $k$ vertices and $\ell$ edges. It \emph{avoids weight $\ell$ on slice $k$} if it does not attain it.
Note that every property attains between $0$ and ${k \choose 2}+1$ weights on slice $k$.

\begin{theorem}[{\cite[Lemma 5.1]{CurticapeanN25}} \& {\cite[Theorem 7.1]{CurticapeanDN025}}]
    \label{thm:hardness-hw}
    Assuming $\ETH$, there are $N_0,\delta > 0$ such that the following holds:
    If $k \geq N_0$ and $0 < d \leq k/2$ and $\Phi$ avoids at least $d \cdot k$ distinct weights and attains at least one weight, all on slice $k$, then no algorithm computes $\numindsub{\slice{\Phi}{k}}{G}$ in time $O(n^{\delta \cdot d})$.
\end{theorem}

For $k \geq \ell + 4$, we can easily determine the number of weights attained by $\Psi_\ell$ on slice $k$:
Since there is no freedom in choosing edges incident to body, tail and sting, this number is precisely $\binom{k - \ell - 2}{2} + 1$.
Thus, the number of weights avoided by  $\Psi_\ell$ on slice $k$ is
\[\binom{k}{2} - \binom{k - \ell - 2}{2} - 1 = k (\ell + 2) - \frac{\ell (\ell + 5)}{2} - 4.\]
For $k \geq 2\ell + 4$, it follows that $\Psi_\ell$ avoids at least $\ell/2 \cdot  k$ weights on slice $k$.
Being nontrivial, it attains at least one weight. 
Theorem~\ref{thm:hardness-hw} readily implies:

\begin{corollary}
    \label{cor:eth-scorpions}
    Assuming $\ETH$, there are $N_0,\delta > 0$ such that the following holds:
    If $k \geq N_0$ and $0 < \ell \leq (k-4)/2$, then no algorithm computes $\numindsub{\slice{\scorp{\ell}}{k}}{G}$ in time $O(n^{\delta \cdot \ell})$.
\end{corollary}

\begin{remark}
    The algorithm from Theorem \ref{thm:alg-scorpion} also demonstrates that the lower bound in Theorem \ref{thm:hardness-hw} cannot be further improved.
    Indeed, $\Psi_\ell$ avoids at most $(\ell + 2) \cdot k$ weights on slice $k$, but $\numindsub{\slice{\Phi}{k}}{G}$ can be solved in $O(n^{\ell+3})$ time.
    Therefore, $\ell$-scorpions show that the bound in Theorem \ref{thm:hardness-hw} is essentially tight.
\end{remark}

With the lower bound of Corollary~\ref{cor:eth-scorpions} at hand, we can obtain properties with varying computational complexity, by choosing $\ell$ dependent on $k$.
For example, setting $\ell \approx \sqrt{k}$ yields a property $\Psi_{\mathrm{sqrt}}$ for which $\pindsub{\Psi_{\mathrm{sqrt}}}$ takes $n^{\Theta(\sqrt{k})}$ time under $\ETH$.
More generally, consider a monotone increasing $f\colon \NN \to \NN$ with $1 \leq f(k) \leq (k-4)/2$ for all $k \in \NN$.
We define $\Psi_f = \bigcup_{k \in \NN} \slice{\Psi_{f(k)}}{k}$ to contain exactly the $k$-vertex $f(k)$-scorpions for all $k \in \NN$.

\begin{corollary}
    \label{cor:eth-scorpions-function}
    Let $f\colon \NN \to \NN$ be monotone increasing with $1 \leq f(k) \leq (k-4)/2$.
    Then $\pindsub{\Psi_f}$ can be solved in time $O(k n^{f(k) + 3})$, and assuming $\ETH$, not in time $O(n^{\alpha \cdot f(k)})$ for a fixed constant $\alpha > 0$.
\end{corollary}

\section{An Updated Conjecture Using the Subgraph Basis}
\label{sec:conj}

Seeing what scorpions have done to Conjecture~\ref{conj:non-meager-hard}, it is natural to update the conjecture and reconsider the question when $\pindsub{\Phi}$ is fixed-parameter tractable.
Towards this, we first observe that our counterexamples can be generalized in a number of ways while keeping their key features. This suggests that a dichotomy theorem for $\pindsub{\Phi}$ will have to encompass various tractable cases, as the following facts indicate:
\begin{itemize}
    \item The key to the efficient algorithm in Theorem \ref{thm:alg-scorpion} is that body, tail and sting are uniquely 
    identified in every $\ell$-scorpion, and that no restrictions are imposed on the subgraph induced by the legs.
    As long as there is a constant-sized set of uniquely identifiable vertices and no restrictions on the remaining vertices, an efficient algorithm for $\pindsub{\Phi}$ follows.
    \item This however is not the final word: Observe that $\pindsub{\Phi}$ has the same complexity as $\pindsub{\neg\Phi}$, where $\neg\Phi$ contains exactly those graphs that are not contained in $\Phi$.
    Notably, the non-scorpion property $\neg\Psi_\ell$ contains the complete graph on $k$ vertices for every $k \geq \ell+4$, which arguably has no constant-sized set of ``uniquely identifiable'' vertices.
    \item Beyond, we can also ``nest'' easy properties:
    For example, let $\Lambda$ be the graph property containing all graphs $H$ with $|V(H)| \geq 9$ so that (a) $H$ is a $1$-scorpion, with some set of legs $X$, and (b) $H[X]$ is not a $2$-scorpion.
    Then $\pindsub{\Lambda}$ can be solved in polynomial time using similar arguments as in Theorem \ref{thm:alg-scorpion}.   
\end{itemize}

The diversity of these examples suggests that a complexity classification for all properties $\Phi$ may be quite intricate.
To obtain a new classification conjecture, we first express the induced subgraph counts arising in $\pindsub{\Phi}$ as linear combinations of (not necessarily induced) subgraph counts. Such basis changes among counting problems have already been used for $\pindsub{\Phi}$ before \cite{CurticapeanDN025,CurticapeanN25,DorflerRSW22,DoringMW24,DoringMW25,FockeR24,RothS20,RothSW24}.%
\footnote{Going further, the problem $\pindsub{\Phi}$ has already been expressed in the similar basis of \emph{homomorphism counts}, and a dichotomy criterion was shown in terms of this basis~\cite{CurticapeanDM17}. However, the transformation into the homomorphism basis renders combinatorial interpretations of $\Phi$ opaque and only yields an implicit dichotomy criterion. In particular, all complexity results listed above were achieved using the \emph{subgraph} basis rather than the homomorphism basis.}

In the following, recall that $\Phi^{(k)}$ for $k\in \mathbb N$ is the restriction of $\Phi$ to $k$-vertex graphs and let $\numindsubstar{\slice{\Phi}{k}}$ be the number of induced $k$-vertex graphs satisfying $\Phi^{(k)}$ in an input graph.
Likewise let $\numsubstar{H}$ denote the number of not necessarily induced $H$-subgraph copies in an input graph.
Then there exists a finite set $\mathcal H$ of unlabeled graphs on exactly $k$ vertices, and coefficients $\alpha_H$ for $H \in \mathcal H$, such that
\begin{equation}
    \label{eq:intro-sub-lincomb}
    \numindsubstar{\slice{\Phi}{k}} = \sum_{H \in \mathcal H}\alpha_{H} \cdot \numsubstar{H}.
\end{equation}
In fact, an explicit formula for the coefficients can be found through inclusion-exclusion:
Given a graph property $\Phi$ and graph $H$, the coefficient $\alpha_H$ in \eqref{eq:intro-sub-lincomb} is given by the so-called \emph{alternating enumerator} $\altEnum{\Phi}{H}$, sometimes defined without the $(-1)^{|E(H)|}$ factor below,
    \[\altEnum{\Phi}{H} \coloneqq (-1)^{|E(H)|} \sum_{S \subseteq E(H)} (-1)^{|S|}\Phi(H[S]).\]

Evaluating fixed finite linear combinations of $k$-vertex graphs as in \eqref{eq:intro-sub-lincomb} is asymptotically no harder than evaluating the individual $H$-subgraph counts with $\altEnum{\Phi}{H} \neq 0$.
For fixed graphs $H$, individual $H$-subgraph counts in turn can be evaluated in time $O(n^{\tau(H)+1})$, where $\tau(H)$ is the vertex-cover number of $H$ (see \cite[Theorem~1.1]{CurticapeanDM17}).
In particular, this implies that $\pindsub{\Phi}$ is fixed-parameter tractable when only graphs $H$ of small vertex-cover number satisfy $\altEnum{\Phi}{H} \neq 0$.
More quantitatively, writing $\tau_\Phi(k)$ for the maximal vertex cover number among $k$-vertex graphs $H$ with $\widehat{\Phi}(H) \neq 0$, we have:

\begin{theorem}[see, e.g., {\cite[Theorem 3.6]{CurticapeanN24}}]
    \label{thm:algorithm-alternating-enumerator-vc}
    For every computable graph property $\Phi$, the problem $\pindsub{\Phi}$ can be solved in time $f(k) \cdot n^{\tau_\Phi(k) + 1}$ for a computable function $f$.
\end{theorem}

Conversely, under certain conditions, a $k$-vertex graph $H$ with $\altEnum{\Phi}{H} \neq 0$ that is hard on its own also implies hardness of the entire linear combination.
For example, under the guarantee that all patterns in the linear combination have the same number $k$ of vertices, this occurs when some graph $H$ with $\altEnum{\Phi}{H} \neq 0$ has large treewidth (see, e.g., \cite{DorflerRSW22}).
Large treewidth is however only a sufficient criterion, since counting $k$-matchings (and more generally, patterns of large vertex-cover number) is hard as well \cite{DBLP:conf/focs/CurticapeanM14}. 

Knowing about a similar phenomenon for the related case of homomorphism counts, one\footnote{Some did, as a superseded arXiv version of a related paper shows: \url{https://arxiv.org/abs/2407.07051v1}} might be led to believe
that linear combinations of $k$-vertex subgraph counts are hard if at least one pattern has large vertex-cover number.
This hypothesis however is too optimistic, as shown by $k$-partial determinants:
These are weighted counts of $k$-vertex cycle covers $C$ in a graph $G$, where $C$ is counted with weight $-1$ if it contains an odd number of cycles, and with weight $1$ otherwise.
The patterns in the subgraph count expansion of $k$-partial determinants are precisely the cycle covers on $k$ vertices. These do have large vertex-cover number, yet $k$-partial determinants can be computed in time $O(n^{\omega+1})$ due to the same cancellations that render usual determinants tractable.

In our updated version of Conjecture~\ref{conj:non-meager-hard}, we assert that such cancellations cannot occur for $\pindsub{\Phi}$. That is, we assert that graphs of large vertex-cover number in the subgraph expansion of $\pindsub{\Phi}$ indeed render the linear combination hard.
Recall that $\tau_\Phi(k)$ denotes the maximal vertex cover number among $k$-vertex graphs $H$ with $\widehat{\Phi}(H) \neq 0$.

\begin{conjecture}
    \label{conj:vc-hard}
    If $\Phi$ is a computable property and the function $\tau_\Phi$ is unbounded, then the problem $\pindsub{\Phi}$ is $\sharpwone$-hard.
\end{conjecture}

Note that, by Theorem \ref{thm:algorithm-alternating-enumerator-vc}, the problem $\pindsub{\Phi}$ is fixed-parameter tractable if $\tau_\Phi$ is bounded.
Hence, Conjecture \ref{conj:vc-hard} formulates a complete characterization in terms of the subgraph basis expression of a property.

To the best of our knowledge, no counterexamples are known for Conjecture~\ref{conj:vc-hard}. Since $k$-partial determinants can be negative, they cannot be written as $\pindsub{\Phi}$ for a property $\Phi$, so they do not form a counterexample to our updated conjecture.
Moreover, we verify that none of the generalized scorpion properties $\Psi_\ell$ for fixed $\ell \in \NN$ refute it.
More specifically, we prove that, for every $k \geq \ell + 4$,
\begin{equation}
    \label{eq:scorpion-vc}
    \tau_{\Psi_\ell}(k) = \ell+2.
\end{equation}
Thus $\tau_{\Psi_\ell}(k) \in O(1)$ for fixed $\ell$, so $\Psi_\ell$ does not satisfy the premises of Conjecture~\ref{conj:vc-hard}.

\subsection{Scorpions in the Subgraph Basis}

We prove \eqref{eq:scorpion-vc} in this section.
Recall the definition of tail, sting, body, and legs of an $\ell$-scorpion
$S$. We call an $\ell$-scorpion $S$ a \emph{skeleton} if its legs form
an independent set. Moreover, an \emph{$\ell$-scorpion fossil} is any graph $S'$ obtained from an $\ell$-scorpion skeleton $S$ by adding an arbitrary number of edges $uv$ with $u,v \in V(S)$ such that at least one of $u,v$ is not a leg.
(Considering Figure~\ref{fig:scorpion}, a scorpion skeleton is obtained by removing all dashed edges. A scorpion fossil is obtained from a scorpion skeleton by adding arbitrary edges, but not between legs.)

To show \eqref{eq:scorpion-vc}, we prove the stronger statement that the graphs $H$ occurring with non-zero coefficients $\altEnum{\Psi_\ell}{H} \neq 0$ in the subgraph expansion of $\Psi_\ell$ are precisely the $\ell$-scorpion fossils.
This is indeed stronger: The vertex-cover number of $\ell$-scorpion fossils is at most $\ell+2$, since they retain the independent set on the legs. 
On the other hand, the augmented biclique $K_{\ell+2,k-\ell-2}^+$ obtained from a complete bipartite graph by turning the left side into a clique is an $\ell$-scorpion fossil of vertex-cover number $\ell+2$.
This graph will also be used in the lower bound under $\SETH$.

\begin{lemma}
    \label{lem:alt-enum-scorp}
    For every $\ell \in \NN$ and every $k$-vertex graph $H$ with $k \geq \ell +4$, we have $\altEnum{\scorp{\ell}}{H} \neq 0$ if and only if $H$ is an $\ell$-scorpion fossil.
\end{lemma}
\begin{proof}
    Let $\CS \coloneqq \{S \subseteq E(H) \mid H[S] \in \scorp{\ell}\}$.
    Similar to Theorem \ref{thm:alg-scorpion}, we use the uniqueness of body, tail, and sting to partition $\CS$ into classes $\CS_\mathbf q$ for $q \in V(H)^{\ell+2}$.
    More specifically, for a tuple $\mathbf{q} = (b, t_1, \dots, t_\ell, s) \in V(H)^{\ell+2}$, let $\CS_\mathbf q$ denote the set of all $S \in \CS$ such that the scorpion $H[S]$ has body $b$, tail $t_1,\ldots,t_\ell$ and sting $s$.
    Note that we may have $\CS_\mathbf q = \emptyset$.
    
    By Lemma~\ref{lem:scorpion:unique}, the sets $\CS_\mathbf q$ for $\mathbf q \in V(H)^{\ell+2}$ partition $\CS$, which implies
    \begin{equation}
        \label{eq:alt-enum-scorp}
        \altEnum{\scorp{\ell}}{H} = (-1)^{|E(H)|} \cdot \sum_{S \in \CS} (-1)^{|S|} = (-1)^{|E(H)|} \cdot \sum_{\mathbf q \in V(H)^{\ell + 2}} \sum_{S \in \CS_\mathbf q} (-1)^{|S|}.
    \end{equation}
    In the following, for a tuple $\mathbf q \in V(H)^{\ell + 2}$, we write $\set(\mathbf q)$ for the set of entries of $\mathbf q$.
    \begin{claim}
        \label{claim:vc-from-alt-enum}
        If $\mathbf q \in V(H)^{\ell + 2}$ is such that $\set(\mathbf q)$ is not a vertex cover of $H$, then
        \[\sum_{S \in \CS_\mathbf q} (-1)^{|S|} = 0.\]
    \end{claim}
    \begin{claimproof}
        Let us first observe that the statement is trivial if $\CS_\mathbf q = \emptyset$.

        Let $E_\mathbf q \coloneqq \{e \in E(H) \mid e \cap \set(\mathbf q) = \emptyset\}$ be the edges of $H$ without an endpoint in $\set(\mathbf q)$.
        Since $\set(\mathbf q)$ is not a vertex cover, we get $E_\mathbf q \neq \emptyset$.
        Observe that $F_\mathbf q \coloneqq S \setminus E_\mathbf q = S' \setminus E_\mathbf q$ for all $S,S' \in \CS_\mathbf q$, because edges incident to body, tail or sting (i.e., edges incident to the vertices in $\mathbf q$) are fixed.
        Since edges between leg vertices can be chosen arbitrarily, we conclude that
        \[\CS_\mathbf q = \{F_\mathbf q \cup E_\mathbf q' \mid E_\mathbf q' \subseteq E_\mathbf q\}.\]
        In particular,
        \[\sum_{S \in \CS_\mathbf q} (-1)^{|S|} = \sum_{E_\mathbf q' \subseteq E_\mathbf q} (-1)^{|F_\mathbf q| + |E_\mathbf q'|} = (-1)^{|F_\mathbf q|}\sum_{E_\mathbf q' \subseteq E_\mathbf q} (-1)^{|E_\mathbf q'|} = 0\]
        since $E_\mathbf q \neq \emptyset$.
    \end{claimproof}

    Now, for the forward direction of the lemma, suppose that $\altEnum{\scorp{\ell}}{H} \neq 0$.
    Then there is some $\mathbf q = (b, t_1, \dots, t_\ell, s)$ such that $\sum_{S \in \CS_\mathbf q} (-1)^{|S|} \neq 0$.
    This implies that
    \begin{enumerate}[label = (\alph*)]
        \item $\set(\mathbf q)$ is a vertex cover of $H$ by Claim \ref{claim:vc-from-alt-enum}, so the vertices in $V(H) \setminus \set(\mathbf q)$ form an independent set in $H$, and
        \item $\CS_\mathbf q \neq \emptyset$, i.e., there is some $S \subseteq E(H)$ such that $H[S]$ is an $\ell$-scorpion with body $b$, tail $t_1,\ldots,t_\ell$ and sting $s$. This means that the vertices in $V(H) \setminus \set(\mathbf q)$ form the legs of $H[S]$.
    \end{enumerate}
    These two properties together directly imply that $H$ is an $\ell$-scorpion fossil.

    For the backward direction, suppose that $H$ is an $\ell$-scorpion fossil, and let $\mathbf{q} = (b, t_1, \dots, t_\ell, s)$ denote the tuple of body, tail, and sting of an underlying $\ell$-scorpion skeleton.
    Then
    \begin{enumerate}[label = (\roman*)]
        \item\label{item:alt-enum-scorp-1} $\CS_\mathbf q \neq \emptyset$, 
        \item\label{item:alt-enum-scorp-2} $\set(\mathbf{q})$ is a vertex cover of $H$.
    \end{enumerate}
    The first item holds because the underlying $\ell$-scorpion skeleton is a witness to $\CS_\mathbf q \neq \emptyset$, while the second item holds from the definition of $\ell$-scorpion skeletons.
    
    Now, let $A \subseteq V(H)^{\ell + 2}$ denote the set of all tuples $\mathbf q'$ satisfying \ref{item:alt-enum-scorp-1} and \ref{item:alt-enum-scorp-2}.
    Then, for every $\mathbf q' \in A$ and every $S \in \mathcal{S}_{\mathbf q'}$, the legs of $H[S]$ form an independent set, since the vertices in $\mathbf q' \in A$ form a vertex-cover. This implies that $H[S]$ is an $\ell$-scorpion skeleton, so $S$ contains exactly $k - 1$ edges, and we obtain
    \begin{equation}\label{eq:goodq}
        \sum_{S \in \CS_{\mathbf q'}} (-1)^{|S|}=(-1)^{k-1} |\CS_{\mathbf q'}| \quad \text{ for all }\mathbf q' \in A.
    \end{equation}
    On the other hand, for $\mathbf q' \in V(H)^{\ell + 2} \setminus A$, the tuple $\mathbf q'$ violates \ref{item:alt-enum-scorp-1}, or $\mathbf q'$ violates \ref{item:alt-enum-scorp-2}. We obtain directly (in the first case) or via Claim~\ref{claim:vc-from-alt-enum} (in the second case) that
    \begin{equation}\label{eq:badq}
        \sum_{S \in \CS_{\mathbf q'}} (-1)^{|S|}=0\quad \text{ for all }\mathbf q' \in V(H)^{\ell + 2}\setminus A.
    \end{equation}
    Let us abbreviate $X_{\mathbf{q}'} \coloneqq \sum_{S \in \CS_{\mathbf q'}} (-1)^{|S|}$.
    Then it follows that
    \begin{align*}
        \altEnum{\scorp{\ell}}{H}
            &= (-1)^{|E(H)|} \sum_{\mathbf q' \in V(H)^{\ell + 2}} X_{\mathbf{q}'}\\
            &= (-1)^{|E(H)|}  \left(\sum_{\substack{\mathbf q' \in V(H)^{\ell + 2} \setminus  A}} X_{\mathbf{q}'} + \sum_{\substack{\mathbf q' \in A}} X_{\mathbf{q}'}\right)\\
            &= (-1)^{|E(H)|}  \left(0 + \sum_{\substack{\mathbf q' \in A}} (-1)^{k-1} |\CS_{\mathbf q'}|\right) = (-1)^{|E(H)| + k - 1} \sum_{\substack{\mathbf q' \in A}} |\CS_{\mathbf q'}| \neq 0,
    \end{align*}
    where we used \eqref{eq:goodq} and \eqref{eq:badq} in the third equality, and where the last inequality holds since $A \neq \emptyset$ and $\CS_{\mathbf q'} \neq \emptyset$ for all $\mathbf q' \in A$.
\end{proof}

As discussed above, the lemma implies Equation~\eqref{eq:scorpion-vc} on the maximum vertex-cover number in the subgraph expansion of the scorpion property.
Together with Theorem \ref{thm:algorithm-alternating-enumerator-vc}, the upper bound on the vertex-cover number gives an alternative proof that $\pindsub{\Psi_\ell}$ is fixed-parameter tractable for $\ell \in \NN$, with the same polynomial degree as in Theorem \ref{thm:alg-scorpion}.
We stress that similar arguments also work for the adaptations of the scorpion property discussed in the beginning of Section \ref{sec:conj}.
For example, for the property $\Lambda$ defined there, we get $\tau_\Lambda(k) \leq 7$ for every $k \in \NN$.

Also, by combining Lemma \ref{lem:alt-enum-scorp} with a result from \cite{BringmannS21}, we obtain the following lower bound under the Strong Exponential-Time Hypothesis $\SETH$.

\begin{proposition}
    \label{prop:seth-scorpions}
    Assuming $\SETH$, for every $\ell \geq 3$ and $\varepsilon > 0$, there is no algorithm solving $\pindsub{\Psi_\ell}$ in time $O(n^{\ell + 2 - \varepsilon})$.
\end{proposition}

\begin{proof}[Proof sketch]
    In the full arXiv version of \cite{BringmannS21}, Bringmann and Slusallek~\cite[Definition~19]{BringmannS21b} define a particular class of bipartite graphs $T(h,0,t)$ on $t+h$ vertices.
    They prove that, for every $t \geq 3$, there is some $h \in \mathbb N$ such that the colorful subgraph decision problem with pattern $T(h,0,t)$ cannot be solved in $O(n^{t -\varepsilon})$ time under $\SETH$; see \cite[Theorem~24.1]{BringmannS21b} and its proof.
    To be specific, for a fixed pattern $H$, the input to this decision problem is a graph $G$ whose vertices are colored with $V(H)$, and the task is to decide whether there is a map $h\colon V(H) \to V(G)$ such that $h(v)$ is $v$-colored for all $v\in V(H)$ and $uv\in E(H)$ implies $h(u)h(v)\in E(G)$.

    This problem with pattern $H$ can be reduced to the problem with pattern $H'\supseteq H$ by adding padding edges to $G$; this increases the size of $G$ by at most a constant factor.
    Since $T(h,0,t)$ is a bipartite graph on $t+h$ vertices, it follows that the colorful subgraph decision problem with the graph $K_{t,h}^+$ as pattern cannot be solved in $O(n^{t -\varepsilon})$ time under $\SETH$.

    By~\cite[Lemma~A.3]{CurticapeanN24}, the colorful subgraph decision problem with pattern $K_{t,h}^+$ and $n$-vertex input graph $G$ can be solved in $O(f(k) \cdot n^2)$ time with an oracle for $\numindsub{\Phi}{\star}$ when $\Phi$ is a property of $k$-vertex graphs with $\altEnum{\Phi}{K_{t,h}^+} \neq 0$.
    Each of the at most $f(k)$ oracle calls queries an uncolored graph on $O(n)$ vertices, so an $O(n^c)$ time algorithm for $\numindsub{\Phi}{\star}$ translates into an $O(n^c)$ time algorithm for the subgraph decision problem for $H'$.

    Combining the above: Colorful subgraph decision with pattern $T(h,0,t)$ cannot be solved in $O(n^{t -\varepsilon})$ time under $\SETH$; the same holds for pattern $K_{t,h}^+$ by padding, and using \cite[Lemma~A.3]{CurticapeanN24} together with \Cref{lem:alt-enum-scorp} gives the same lower bound for $\pindsub{\scorp{\ell}}$ with $t=\ell+2$.
\end{proof}

\begin{remark}
    We stress that, in contrast to \Cref{cor:eth-scorpions,cor:eth-scorpions-function}, the $\SETH$ lower bound from Proposition \ref{prop:seth-scorpions} only holds for fixed $\ell$, and not if $\ell$ is chosen as a function of $k$.
\end{remark}

\subsection{Implications of the Conjecture}

To conclude the paper and as a potential avenue for future work, we show that our refined Conjecture~\ref{conj:vc-hard} would easily imply generalizations and variants of known hardness results.

Most of the recent works (see, e.g., \cite{CurticapeanN25,DorflerRSW22,DoringMW24,DoringMW25,RothSW24}) prove hardness of $\pindsub{\Phi}$ by identifying graphs with non-zero alternating enumerator and unbounded \emph{treewidth} rather than vertex-cover number; this is a stronger property.
More concretely, let $\eta_\Phi(k)$ denote the maximal treewidth among $k$-vertex graphs $H$ with $\widehat{\Phi}(H) \neq 0$.
If $\Phi$ is computable and $\eta_\Phi$ is unbounded, then $\pindsub{\Phi}$ is $\sharpwone$-hard (see, e.g., \cite[Theorem 3.1]{CurticapeanN25}).
However, showing that $\eta_\Phi$ is unbounded is often quite a challenge.
Conjecture~\ref{conj:vc-hard} postulates that it suffices for $\tau_\Phi$ to be unbounded to obtain $\sharpwone$-hardness, which is often much easier to show.

As a concrete example, previous work \cite{DorflerRSW22} established hardness of $\pindsub{\Phi}$ in the presence of infinitely many primes $p$ with $\Phi(\IS_{2p}) \neq \Phi(K_{p,p})$, where $\IS_{2p}$ and $K_{p,p}$ are the edge-less and complete bipartite graphs, respectively.
Similarly, for edge-monotone properties $\Phi$, which are studied in \cite{CurticapeanN25,DoringMW24,DoringMW25}, we typically obtain $\Phi(\IS_k) \neq \Phi(K_k)$ for all sufficiently large $k$.
As we show below, Conjecture~\ref{conj:vc-hard} implies that infinitely many \emph{arbitrary} prime-order vertex-transitive graphs $H_1$, $H_2$ with $\Phi(H_1) \neq \Phi(H_2)$ suffice for hardness. This can be seen as a generalization of \cite{DorflerRSW22}, since $K_{p, p}$ and $\IS_{2p}$ both contain transitive $p$-groups in their automorphism groups. 

\begin{proposition}
    Let $\Phi$ be a graph property such that, for infinity many primes $p$, there are vertex-transitive graphs $H_1$, $H_2$ with $p$ vertices and $\Phi(H_1) \neq \Phi(H_2)$.
    Then $\pindsub{\Phi}$ is $\sharpwone$-hard assuming Conjecture \ref{conj:vc-hard}.
\end{proposition}
\begin{proof}
    For a prime $p$, we say that a group $\Gamma$ is a $p$-group if its order is a power of $p$, and we establish some facts about groups and graphs:
    \begin{enumerate}[label = (\Roman*)]
        \item\label{item:p-subgroup} Each transitive group that operates on a set of $p$ elements contains a transitive $p$-subgroup. This holds since $\Gamma$ contains a cyclic group of order $p$, which is transitive (see the discussion before \cite[Theorem 1]{Neumann74} for more details).
        \item\label{item:transitive-regular} Each vertex-transitive graph is regular. This is trivial.
        \item\label{item:regular-vc} Every $d$-regular $k$-vertex graph $H$ with $d>1$ has vertex-cover number $\tau(H) \geq k/2 - 1$:
        By \cite{Rosenfeld64}, the independence number $\alpha(H)$ of $H$ is at most $k/2 + 1$, so $\tau(H) = k - \alpha(H) \geq k/2 - 1$.
    \end{enumerate}

    By the requirements of the proposition, there is a vertex-transitive graph $H_p$ with $p$ vertices and $\Phi(H_p) \neq \Phi(\IS_p)$ for infinitely many primes $p$.
    Let $H_p$ be such a graph.
    Then $\Aut(H_p)$ contains a transitive $p$-subgroup $\Gamma_p$ by \ref{item:p-subgroup}.
    We say a subgraph $F_p$ of $H_p$ is a \emph{fixed point of $\Gamma_p$ in $H_p$} if $V(F_p) = V(H_p)$ and $\gamma(E(F_p)) = E(F_p)$ for all $\gamma \in \Gamma_p$, where $\gamma(E(F_p)) \coloneqq \{\gamma(u)\gamma(v) \mid uv \in E(F_p)\}$.
    Let $\fp{\Gamma_p}{H_p}$ denote the set of fixed points of $\Gamma_p$ in $H_p$ (see \cite[Appendix A]{DoringMW25} for more details on fixed points of group actions in graphs). 

    We show that there is a fixed point $F_p \in \fp{\Gamma_p}{H_p}$ that has a non-zero alternating enumerator and $\tau(F_p) \geq p/2 - 1$.
    For this, let $F_p$ be a fixed point in $\fp{\Gamma_p}{H_p}$ with $\Phi(F_p) \neq \Phi(\IS_p)$ and $\Phi(F'_p) = \Phi(\IS_p)$ for all $F'_p \in \fp{\Gamma_p}{H_p}$ with $F'_p \subsetneq F_p$ (i.e., $F'_p$ is a proper subgraph of $F_p$).
    Such a graph $F_p$ exists since $H_p \in \fp{\Gamma_p}{H_p}$ and $\Phi(H_p) \neq \Phi(\IS_p)$.
    Now, \cite[Lemma~4.8]{DoringMW24} implies that $\altEnum{\Phi}{F_p} \neq 0$.
    Moreover, \cite[Lemma~3.1]{DoringMW24} implies that $\Gamma_p \subseteq \Aut(F_p)$, and hence $F_p$ is vertex-transitive.
    Lastly, \ref{item:transitive-regular} and \ref{item:regular-vc} yield that $\tau(F_p) \geq p/2 - 1$.

    This means that we can find a sequence of graphs $F_p$ with unbounded vertex cover number and $\widehat{\Phi}(F_p) \neq 0$, and thus $\pindsub{\Phi}$ is $\sharpwone$-hard by \Cref{conj:vc-hard}.
\end{proof}

\bibliographystyle{plainurl}
\bibliography{refs}

\begin{thebibliography}{10}

\bibitem{DBLP:conf/soda/AlmanDWXXZ25}
Josh Alman, Ran Duan, Virginia {Vassilevska Williams}, Yinzhan Xu, Zixuan Xu, and Renfei Zhou.
\newblock More asymmetry yields faster matrix multiplication.
\newblock In Yossi Azar and Debmalya Panigrahi, editors, {\em Proceedings of the 2025 Annual {ACM-SIAM} Symposium on Discrete Algorithms, {SODA} 2025, New Orleans, LA, USA, January 12-15, 2025}, pages 2005--2039. {SIAM}, 2025.
\newblock \href {https://doi.org/10.1137/1.9781611978322.63} {\path{doi:10.1137/1.9781611978322.63}}.

\bibitem{BringmannS21}
Karl Bringmann and Jasper Slusallek.
\newblock Current algorithms for detecting subgraphs of bounded treewidth are probably optimal.
\newblock In Nikhil Bansal, Emanuela Merelli, and James Worrell, editors, {\em 48th International Colloquium on Automata, Languages, and Programming, {ICALP} 2021, July 12-16, 2021, Glasgow, Scotland (Virtual Conference)}, volume 198 of {\em LIPIcs}, pages 40:1--40:16. Schloss Dagstuhl - Leibniz-Zentrum f{\"{u}}r Informatik, 2021.
\newblock \href {https://doi.org/10.4230/LIPICS.ICALP.2021.40} {\path{doi:10.4230/LIPICS.ICALP.2021.40}}.

\bibitem{BringmannS21b}
Karl Bringmann and Jasper Slusallek.
\newblock Current algorithms for detecting subgraphs of bounded treewidth are probably optimal.
\newblock {\em CoRR}, abs/2105.05062, 2021.
\newblock \href {https://arxiv.org/abs/2105.05062} {\path{arXiv:2105.05062}}.

\bibitem{ChenHKX06}
Jianer Chen, Xiuzhen Huang, Iyad~A. Kanj, and Ge~Xia.
\newblock Strong computational lower bounds via parameterized complexity.
\newblock {\em J. Comput. Syst. Sci.}, 72(8):1346--1367, 2006.
\newblock \href {https://doi.org/10.1016/J.JCSS.2006.04.007} {\path{doi:10.1016/J.JCSS.2006.04.007}}.

\bibitem{CurticapeanDM17}
Radu Curticapean, Holger Dell, and D{\'{a}}niel Marx.
\newblock Homomorphisms are a good basis for counting small subgraphs.
\newblock In Hamed Hatami, Pierre McKenzie, and Valerie King, editors, {\em Proceedings of the 49th Annual {ACM} {SIGACT} Symposium on Theory of Computing, {STOC} 2017, Montreal, QC, Canada, June 19-23, 2017}, pages 210--223. {ACM}, 2017.
\newblock \href {https://doi.org/10.1145/3055399.3055502} {\path{doi:10.1145/3055399.3055502}}.

\bibitem{CurticapeanDN025}
Radu Curticapean, Simon D{\"{o}}ring, Daniel Neuen, and Jiaheng Wang.
\newblock Can you link up with treewidth?
\newblock In Olaf Beyersdorff, Michal Pilipczuk, Elaine Pimentel, and Kim~Thang Nguyen, editors, {\em 42nd International Symposium on Theoretical Aspects of Computer Science, {STACS} 2025, March 4-7, 2025, Jena, Germany}, volume 327 of {\em LIPIcs}, pages 28:1--28:24. Schloss Dagstuhl - Leibniz-Zentrum f{\"{u}}r Informatik, 2025.
\newblock \href {https://doi.org/10.4230/LIPICS.STACS.2025.28} {\path{doi:10.4230/LIPICS.STACS.2025.28}}.

\bibitem{DBLP:conf/focs/CurticapeanM14}
Radu Curticapean and D{\'{a}}niel Marx.
\newblock Complexity of counting subgraphs: Only the boundedness of the vertex-cover number counts.
\newblock In {\em 55th {IEEE} Annual Symposium on Foundations of Computer Science, {FOCS} 2014, Philadelphia, PA, USA, October 18-21, 2014}, pages 130--139. {IEEE} Computer Society, 2014.
\newblock \href {https://doi.org/10.1109/FOCS.2014.22} {\path{doi:10.1109/FOCS.2014.22}}.

\bibitem{CurticapeanN24}
Radu Curticapean and Daniel Neuen.
\newblock Counting small induced subgraphs: Hardness via fourier analysis.
\newblock {\em CoRR}, abs/2407.07051, 2024.
\newblock \href {https://arxiv.org/abs/2407.07051} {\path{arXiv:2407.07051}}.

\bibitem{CurticapeanN25}
Radu Curticapean and Daniel Neuen.
\newblock Counting small induced subgraphs: Hardness via {F}ourier analysis.
\newblock In Yossi Azar and Debmalya Panigrahi, editors, {\em Proceedings of the 2025 {ACM-SIAM} Symposium on Discrete Algorithms, {SODA} 2025, New Orleans, LA, USA, January 12-15, 2025}, pages 3677--3695. {SIAM}, 2025.
\newblock \href {https://doi.org/10.1137/1.9781611978322.122} {\path{doi:10.1137/1.9781611978322.122}}.

\bibitem{CyganFKLMPPS15}
Marek Cygan, Fedor~V. Fomin, Lukasz Kowalik, Daniel Lokshtanov, D{\'{a}}niel Marx, Marcin Pilipczuk, Michal Pilipczuk, and Saket Saurabh.
\newblock {\em Parameterized Algorithms}.
\newblock Springer, 2015.
\newblock \href {https://doi.org/10.1007/978-3-319-21275-3} {\path{doi:10.1007/978-3-319-21275-3}}.

\bibitem{Diestel17}
Reinhard Diestel.
\newblock {\em Graph Theory}.
\newblock Springer Berlin, 5 edition, 2017.
\newblock \href {https://doi.org/10.1007/978-3-662-53622-3} {\path{doi:10.1007/978-3-662-53622-3}}.

\bibitem{DorflerRSW22}
Julian D{\"{o}}rfler, Marc Roth, Johannes Schmitt, and Philip Wellnitz.
\newblock Counting induced subgraphs: An algebraic approach to {\#}{W}[1]-hardness.
\newblock {\em Algorithmica}, 84(2):379--404, 2022.
\newblock \href {https://doi.org/10.1007/S00453-021-00894-9} {\path{doi:10.1007/S00453-021-00894-9}}.

\bibitem{DoringMW24}
Simon D{\"{o}}ring, D{\'{a}}niel Marx, and Philip Wellnitz.
\newblock Counting small induced subgraphs with edge-monotone properties.
\newblock In Bojan Mohar, Igor Shinkar, and Ryan O'Donnell, editors, {\em Proceedings of the 56th Annual {ACM} Symposium on Theory of Computing, {STOC} 2024, Vancouver, BC, Canada, June 24-28, 2024}, pages 1517--1525. {ACM}, 2024.
\newblock \href {https://doi.org/10.1145/3618260.3649644} {\path{doi:10.1145/3618260.3649644}}.

\bibitem{DoringMW25}
Simon D{\"{o}}ring, D{\'{a}}niel Marx, and Philip Wellnitz.
\newblock From graph properties to graph parameters: Tight bounds for counting on small subgraphs.
\newblock In Yossi Azar and Debmalya Panigrahi, editors, {\em Proceedings of the 2025 {ACM-SIAM} Symposium on Discrete Algorithms, {SODA} 2025, New Orleans, LA, USA, January 12-15, 2025}, pages 3637--3676. {SIAM}, 2025.
\newblock \href {https://doi.org/10.1137/1.9781611978322.121} {\path{doi:10.1137/1.9781611978322.121}}.

\bibitem{FlumG06}
J{\"{o}}rg Flum and Martin Grohe.
\newblock {\em Parameterized Complexity Theory}.
\newblock Texts in Theoretical Computer Science. An {EATCS} Series. Springer, 2006.
\newblock \href {https://doi.org/10.1007/3-540-29953-X} {\path{doi:10.1007/3-540-29953-X}}.

\bibitem{FockeR24}
Jacob Focke and Marc Roth.
\newblock Counting small induced subgraphs with hereditary properties.
\newblock {\em {SIAM} J. Comput.}, 53(2):189--220, 2024.
\newblock \href {https://doi.org/10.1137/22M1512211} {\path{doi:10.1137/22M1512211}}.

\bibitem{GoldbergR24}
Leslie~Ann Goldberg and Marc Roth.
\newblock Parameterised and fine-grained subgraph counting, modulo 2.
\newblock {\em Algorithmica}, 86(4):944--1005, 2024.
\newblock \href {https://doi.org/10.1007/S00453-023-01178-0} {\path{doi:10.1007/S00453-023-01178-0}}.

\bibitem{DBLP:conf/stoc/GroheSS01}
Martin Grohe, Thomas Schwentick, and Luc Segoufin.
\newblock When is the evaluation of conjunctive queries tractable?
\newblock In Jeffrey~Scott Vitter, Paul~G. Spirakis, and Mihalis Yannakakis, editors, {\em Proceedings on 33rd Annual {ACM} Symposium on Theory of Computing, July 6-8, 2001, Heraklion, Crete, Greece}, pages 657--666. {ACM}, 2001.
\newblock \href {https://doi.org/10.1145/380752.380867} {\path{doi:10.1145/380752.380867}}.

\bibitem{IP01}
Russell Impagliazzo and Ramamohan Paturi.
\newblock On the complexity of k-{SAT}.
\newblock {\em J. Comput. Syst. Sci.}, 62(2):367--375, 2001.
\newblock \href {https://doi.org/10.1006/JCSS.2000.1727} {\path{doi:10.1006/JCSS.2000.1727}}.

\bibitem{IPZ01}
Russell Impagliazzo, Ramamohan Paturi, and Francis Zane.
\newblock Which problems have strongly exponential complexity?
\newblock {\em J. Comput. Syst. Sci.}, 63(4):512--530, 2001.
\newblock \href {https://doi.org/10.1006/JCSS.2001.1774} {\path{doi:10.1006/JCSS.2001.1774}}.

\bibitem{DBLP:journals/siamcomp/ItaiR78}
Alon Itai and Michael Rodeh.
\newblock Finding a minimum circuit in a graph.
\newblock {\em {SIAM} J. Comput.}, 7(4):413--423, 1978.
\newblock \href {https://doi.org/10.1137/0207033} {\path{doi:10.1137/0207033}}.

\bibitem{JerrumM15}
Mark Jerrum and Kitty Meeks.
\newblock The parameterised complexity of counting connected subgraphs and graph motifs.
\newblock {\em J. Comput. Syst. Sci.}, 81(4):702--716, 2015.
\newblock \href {https://doi.org/10.1016/J.JCSS.2014.11.015} {\path{doi:10.1016/J.JCSS.2014.11.015}}.

\bibitem{JerrumM15b}
Mark Jerrum and Kitty Meeks.
\newblock Some hard families of parameterized counting problems.
\newblock {\em {ACM} Trans. Comput. Theory}, 7(3):11:1--11:18, 2015.
\newblock \href {https://doi.org/10.1145/2786017} {\path{doi:10.1145/2786017}}.

\bibitem{DBLP:conf/icml/JinBCL24}
Emily Jin, Michael~M. Bronstein, {\.I}smail~{\.I}lkan Ceylan, and Matthias Lanzinger.
\newblock Homomorphism counts for graph neural networks: All about that basis.
\newblock In {\em Forty-first International Conference on Machine Learning, {ICML} 2024, Vienna, Austria, July 21-27, 2024}. OpenReview.net, 2024.
\newblock URL: \url{https://openreview.net/forum?id=zRrzSLwNHQ}.

\bibitem{Kozlov08}
Dmitry~N. Kozlov.
\newblock {\em Combinatorial Algebraic Topology}, volume~21 of {\em Algorithms and computation in mathematics}.
\newblock Springer, 2008.
\newblock \href {https://doi.org/10.1007/978-3-540-71962-5} {\path{doi:10.1007/978-3-540-71962-5}}.

\bibitem{LenstraBE74}
Hendrik Lenstra, Marc~R. Best, and Peter van Emde~Boas.
\newblock A sharpened version of the {A}anderaa-{R}osenberg conjecture.
\newblock {\em Report 30/74, Mathematisch Centrum Amsterdam (1974)}, pages 1--20, 1974.
\newblock URL: \url{https://hdl.handle.net/1887/3792}.

\bibitem{Milo02}
Ron Milo, Shai~S. Shen-Orr, Shalev Itzkovitz, Nadav Kashtan, Dmitri~B. Chklovskii, and Uri Alon.
\newblock Network motifs: Simple building blocks of complex networks.
\newblock {\em Science}, 298(5594):824--827, 2002.
\newblock \href {https://doi.org/10.1126/science.298.5594.824} {\path{doi:10.1126/science.298.5594.824}}.

\bibitem{Neumann74}
Peter~M. Neumann.
\newblock Transitive permutation groups of prime degree.
\newblock In {\em Proceedings of the Second International Conference on the Theory of Groups}, volume Vol. 372 of {\em Lecture Notes in Math.}, pages 520--535. Springer Berlin Heidelberg, 1974.
\newblock \href {https://doi.org/10.1007/978-3-662-21571-5\_55} {\path{doi:10.1007/978-3-662-21571-5\_55}}.

\bibitem{NesetrilP85}
Jaroslav Ne\v{s}et\v{r}il and Svatopluk Poljak.
\newblock On the complexity of the subgraph problem.
\newblock {\em Comment. Math. Univ. Carolin.}, 26(2):415--419, 1985.
\newblock URL: \url{http://dml.cz/dmlcz/106381}.

\bibitem{DBLP:journals/sigact/Rosenberg73}
Arnold~L. Rosenberg.
\newblock On the time required to recognize properties of graphs: a problem.
\newblock {\em {SIGACT} News}, 5(4):15--16, 1973.
\newblock \href {https://doi.org/10.1145/1008299.1008302} {\path{doi:10.1145/1008299.1008302}}.

\bibitem{Rosenfeld64}
Moshe Rosenfeld.
\newblock Independent sets in regular graphs.
\newblock {\em Israel J. Math.}, 2:262--272, 1964.
\newblock \href {https://doi.org/10.1007/BF02759743} {\path{doi:10.1007/BF02759743}}.

\bibitem{RothS20}
Marc Roth and Johannes Schmitt.
\newblock Counting induced subgraphs: {A} topological approach to {\#}{W}[1]-hardness.
\newblock {\em Algorithmica}, 82(8):2267--2291, 2020.
\newblock \href {https://doi.org/10.1007/S00453-020-00676-9} {\path{doi:10.1007/S00453-020-00676-9}}.

\bibitem{DBLP:conf/icalp/Roth0W21}
Marc Roth, Johannes Schmitt, and Philip Wellnitz.
\newblock Detecting and counting small subgraphs, and evaluating a parameterized tutte polynomial: Lower bounds via toroidal grids and cayley graph expanders.
\newblock In Nikhil Bansal, Emanuela Merelli, and James Worrell, editors, {\em 48th International Colloquium on Automata, Languages, and Programming, {ICALP} 2021, July 12-16, 2021, Glasgow, Scotland (Virtual Conference)}, volume 198 of {\em LIPIcs}, pages 108:1--108:16. Schloss Dagstuhl - Leibniz-Zentrum f{\"{u}}r Informatik, 2021.
\newblock \href {https://doi.org/10.4230/LIPICS.ICALP.2021.108} {\path{doi:10.4230/LIPICS.ICALP.2021.108}}.

\bibitem{RothSW24}
Marc Roth, Johannes Schmitt, and Philip Wellnitz.
\newblock Counting small induced subgraphs satisfying monotone properties.
\newblock {\em {SIAM} J. Comput.}, 53(6):FOCS20--139--–FOCS20--174, 2024.
\newblock \href {https://doi.org/10.1137/20M1365624} {\path{doi:10.1137/20M1365624}}.

\bibitem{SchillerJHS15}
Benjamin Schiller, Sven Jager, Kay Hamacher, and Thorsten Strufe.
\newblock {StreaM} - {A} stream-based algorithm for counting motifs in dynamic graphs.
\newblock In Adrian{-}Horia Dediu, Francisco~Hern{\'{a}}ndez Quiroz, Carlos Mart{\'{\i}}n{-}Vide, and David~A. Rosenblueth, editors, {\em Algorithms for Computational Biology - Second International Conference, AlCoB 2015, Mexico City, Mexico, August 4-5, 2015, Proceedings}, volume 9199 of {\em Lecture Notes in Computer Science}, pages 53--67. Springer, 2015.
\newblock \href {https://doi.org/10.1007/978-3-319-21233-3\_5} {\path{doi:10.1007/978-3-319-21233-3\_5}}.

\bibitem{DBLP:journals/tcsb/SchreiberS05}
Falk Schreiber and Henning Schw{\"{o}}bbermeyer.
\newblock Frequency concepts and pattern detection for the analysis of motifs in networks.
\newblock {\em Trans. Comp. Sys. Biology}, 3:89--104, 2005.
\newblock \href {https://doi.org/10.1007/11599128\_7} {\path{doi:10.1007/11599128\_7}}.

\end{thebibliography}

\end{document}